\documentclass{article}
\usepackage{amsmath}
\usepackage{amsthm}
\usepackage{amsfonts}
\usepackage{amssymb}
\usepackage{amscd}
\usepackage[mathscr]{eucal}
\usepackage{mathtools}
\usepackage{fullpage}
\usepackage{color}
\usepackage{float}
\usepackage{graphicx}
\newtheorem{theorem}{Theorem}
\newtheorem{proposition}{Proposition}
\newtheorem{lemma}{Lemma}
\newtheorem{corollary}{Corollary}
\newcommand*{\myabstract}[1]{ \begin{abstract} #1 \end{abstract} }
\begin{document}
\title{\bf On commutativity and near commutativity of translational and rotational averages: Analytical proofs and numerical examinations}
\author{Len Bos%
\footnote{Dipartimento di Informatica, Universit\`a di Verona, Italy;
{\tt leonardpeter.bos@univr.it}}\,,
David R. Dalton%
\footnote{Department of Earth Sciences, Memorial University of Newfoundland,  Canada; 
{\tt dalton.nfld@gmail.com}}\,,
Michael A. Slawinski%
\footnote{Department of Earth Sciences, 
Memorial University of Newfoundland, Canada; 
{\tt mslawins@mac.com}}}
\maketitle
\myabstract{%
We show that, in general, the translational average over a spatial variable---discussed by Backus~\cite{backus}, and referred to as the equivalent-medium average---and the rotational average over a symmetry group at a point---discussed by Gazis et al.~\cite{gazis}, and referred to as the effective-medium average---do not commute.
However, they do commute in special cases of particular symmetry classes, which correspond to special relations among the elasticity parameters.
We also show that this noncommutativity is a function of the strength of anisotropy.
Surprisingly, a perturbation of the elasticity parameters about a point of weak anisotropy results in the commutator of the two types of averaging being of the order of the {\it square\/} of this perturbation.
Thus, these averages nearly commute in the case of weak anisotropy, which is of interest in such disciplines as quantitative seismology, where the weak-anisotropy assumption results in empirically adequate models.}
\section{Introduction}
Hookean solids are defined by their mechanical property relating linearly the stress tensor,~$\sigma$\,, and the strain tensor,~$\varepsilon$\,,
\begin{equation*}
\sigma_{ij}=\sum_{k=1}^3\sum_{\ell=1}^3c_{ijk\ell}\varepsilon_{k\ell}\,,\qquad i,j=1,2,3
\,.
\end{equation*}
The elasticity tensor,~$c$\,, belongs to one of the eight material-symmetry classes shown in Figure~\ref{fig:orderrelation}.
\begin{figure}
\begin{center}
\includegraphics[scale=0.7]{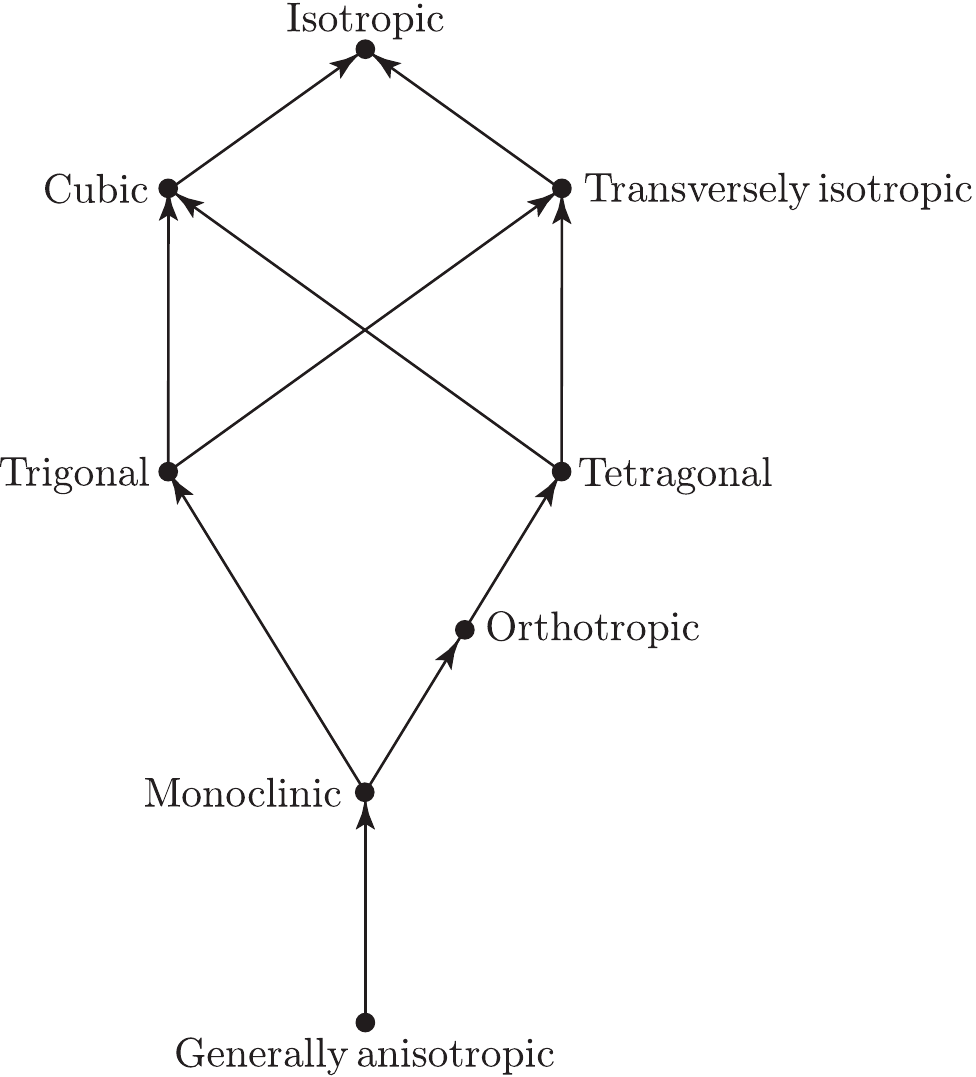}
\end{center}
\caption{\small{Partial ordering of material-symmetry classes of elasticity tensors:
Arrows indicate subgroups. For instance, monoclinic is a subgroup of all symmetries, except general anisotropy; in particular, it is a subgroup of both orthotropic and trigonal symmetries, but orthotropic symmetry is not a subgroup of trigonal or {\it vice-versa}.}}
\label{fig:orderrelation}
\end{figure}

The Backus~\cite{backus} average, which is a moving average over a spatial inhomogeneity, allows us to quantify the response of a wave propagating through a series of parallel layers whose thicknesses are much smaller than the wavelength of a signal.
Each layer is a homogeneous Hookean solid exhibiting a given material symmetry with its elasticity parameters.
The average results in a Hookean solid whose elasticity parameters---and, hence, its material symmetry---allow us to model a long-wavelength response.
The material symmetry of a resulting medium, which we refer to as {\sl equivalent}, is a consequence of the symmetries exhibited by the averaged layers.

As shown by Backus~\cite{backus}, the medium equivalent to a stack of isotropic or transversely isotropic layers is a homogeneous, or nearly homogeneous, transversely is\-o\-trop\-ic medium, where a {\it nearly\/} homogeneous medium is a consequence of a {\em moving\/} average.
The Backus~\cite{backus} formulation is reviewed and extended by Bos et al.~\cite{bos}, where formulations for generally anisotropic, monoclinic, and orthotropic thin layers are also derived.
Also, Bos et al.~\cite{bos} examine the underlying assumptions and approximations behind the Backus~\cite{backus} formulation, which is derived by expressing rapidly varying stresses and strains in terms of products of algebraic combinations of rapidly varying elasticity parameters with slowly varying stresses and strains.
The only mathematical approximation of Backus~\cite{backus} is that the average of a product of a rapidly varying function and a slowly varying function is approximately equal to the product of the averages of these two functions.
This approximation is discussed by Bos et al.~\cite{bos,BosX}.

According to Backus~\cite{backus}, the average of $f(x_3)$ of ``width''~$\ell'$  is
\begin{equation}
\label{eq:BackusOne}
\overline f(x_3):=\int\limits_{-\infty}^\infty w(\zeta-x_3)f(\zeta)\,{\rm d}\zeta
\,,
\end{equation}
where $w(x_3)$ is a weight function with the following properties:
\begin{equation*}
w(x_3)\geqslant0\,,
\quad w(\pm\infty)=0\,,
\quad
\int\limits_{-\infty}^\infty w(x_3)\,{\rm d}x_3=1\,,
\end{equation*}
\vspace*{-10pt}
\begin{equation*}
\int\limits_{-\infty}^\infty x_3w(x_3)\,{\rm d}x_3=0\,,
\quad
\int\limits_{-\infty}^\infty x_3^2w(x_3)\,{\rm d}x_3=(\ell')^2\,.
\end{equation*}
These properties define $w(x_3)$ as a probability-density function, whose mean is zero and whose standard deviation is~$\ell'$\,, thus explaining the use of the term ``width'' for~$\ell'$\,.

The Gazis et al.~\cite{gazis} average, which is an average over an anisotropic symmetry group, allows us to obtain the closest symmetric counterpart---in the Frobenius sense---of a chosen material symmetry to a generally anisotropic Hookean solid.
The average is a Hookean solid, to which we refer as {\sl effective}, and whose elasticity parameters correspond to a symmetry chosen {\it a priori}.

The Gazis et al.~\cite{gazis} average is a projection given by
\begin{equation}
\widetilde c^{\,\,\rm sym}:=\intop_{G^{\rm sym}}(g\circ c)\,\mathrm{d}\mu(g)
\,,
\label{eq:proj}
\end{equation}
where the integration is over the symmetry group,~$G^{\rm sym}$\,, whose elements are~$g$\,, with respect to the invariant measure,~$\mu$\,, normalized so that $\mu(G^{\rm sym})=1$\,; $\widetilde c^{\,\,\rm sym}$ is the orthogonal projection of $c$\,, in the sense of the Frobenius norm,  onto the linear space containing all tensors of that symmetry, which are~$c^{\,\,\rm sym}$\,.
Integral~(\ref{eq:proj}) reduces to a finite sum for the classes whose symmetry groups are finite, which are all classes in Figure~\ref{fig:orderrelation}, except isotropy and transverse isotropy. 

The Gazis et al.~\cite{gazis} approach is reviewed and extended by Danek et al.~\cite{dks1,dks2} in the context of random errors.
Therein, elasticity tensors are not constrained to the same---or even different but known---orientation of the coordinate system.
In other words, in general, the closest---and more symmetric counterpart---exhibits different orientation of symmetry planes and axes than does its original material. 

Let us emphasize that the fundamental distinction between the two averages is their domain of operation.
The Gazis et al.~\cite{gazis} average is an average over symmetry groups at a point and the Backus~\cite{backus} average is a spatial average over a distance.
These averages can be used separately or together.
Hence, an examination of their commutativity provides us with an insight into their meaning and into allowable mathematical operations.

The interplay between anisotropy and inhomogeneity is an important factor in modelling traveltime data in seismology.
Similar traveltimes can be obtained by considering anisotropy,
inhomogeneity or their combination.
However---since the purpose of modelling is to
infer a realistic medium, not only to account for the measured traveltimes---the
interplay between anisotropy and inhomogeneity is investigated in the context of
symmetry increase, homogenization and their commutativity.

The commutator of two operators is defined as $[A,B]:=AB-BA$ and is zero if $A$ and $B$ commute; more generally, the size of the commutator gives an indication of how close they are to commuting.
In our case, we apply the two types of averages to a medium with a certain symmetry class with parameters that may be perturbed by a perturbation parameter, say,~$h$\,.
Thus we may consider the commutator $[A,B]=:F(h)$ to be a function of $h$\,.
If, for no perturbation---which means that $h=0$\,---the averages commute---in other words, $F(0)=0$\,---we expect
\begin{equation*}
[A,B]=F(0)+F'(0)\,h + \cdots=F'(0)\,h+\cdots	
\end{equation*}
to be of order $h$\,.
Surprisingly, we show that in certain cases, perturbing about a symmetry class for which there is commutativity, we have $[A,B]=O(h^2)$\,, which means that the commutator is much smaller than might originally have been expected, and we have very near commutativity.

We begin this paper by formulating analytically the commutativity diagrams between the two averages.
We proceed from generally anisotropic layers to a monoclinic medium, from monoclinic layers to an orthotropic medium, and from orthotropic layers to a tetragonal medium.
Also, we discuss transversely isotropic layers, which---depending on the order of operations---result in a transversely isotropic or isotropic medium.
Subsequently, we examine numerically the commutativity, which allows us to consider the case of weak anisotropy.
We conclude this paper with both expected and unexpected results.
\section{Analytical formulation}
\label{sec:AnalForm}
\subsection{Generally anisotropic layers and monoclinic medium}
Let us consider a stack of generally anisotropic layers to obtain a monoclinic medium.
To examine the commutativity between the Backus~\cite{backus} and Gazis et al.~\cite{gazis} averages, let us study the following diagram,
\begin{equation}
\label{eq:CD2}
\begin{CD}
\rm{aniso}@>\rm{B}>>\rm{aniso}\\
@V\mathrm{G}VV                         @VV\rm{G}V\\
\rm{mono}@>>\rm{B}>\rm{mono}
\end{CD}
\end{equation}
and  Theorem~\ref{thm:One}, as well as its corollary.
\begin{theorem}
\label{thm:One}
In general, the Backus~\cite{backus} and Gazis et al.~\cite{gazis} averages do not commute.
\end{theorem}
\begin{proof}
This is a consequence of  the following  more specific case.
\begin{proposition}
For the generally anisotropic and monoclinic symmetries, the Backus~\cite{backus} and Gazis et al.~\cite{gazis} averages do not commute.
\end{proposition}
\noindent To understand this corollary, we invoke the following lemma, whose proof is in Appendix~\ref{AppOne1}.
\begin{lemma}
\label{lem:Mono}
For the effective monoclinic symmetry, the result of the Gazis et al.~\cite{gazis} average is tantamount to replacing each $c_{ijk\ell}$\,, in a generally anisotropic tensor, by its corresponding $c_{ijk\ell}$ of the monoclinic tensor, expressed in the natural coordinate system, including replacements of the anisotropic-tensor components by the zeros of the corresponding monoclinic components. 
\end{lemma}

Let us first examine the counterclockwise path of Diagram~(\ref{eq:CD2}).
Lemma~\ref{lem:Mono} entails the following corollary.
\begin{corollary}
\label{col:Mono}
For the effective monoclinic symmetry, given a generally anisotropic tensor,~$C$\,, 
\begin{equation}
\label{eq:GazisMono}
\widetilde{C}^{\,\rm mono}=C^{\,\rm mono}
\,;
\end{equation} 
where $\widetilde{C}^{\,\rm mono}$ is the Gazis et al.~\cite{gazis} average of~$C$\,, and $C^{\,\rm mono}$ is the monoclinic tensor whose nonzero entries are the same as for~$C$\,.
\end{corollary}

According to Corollary~\ref{col:Mono}, the effective monoclinic
tensor is obtained simply by setting  to zero---in the generally anisotropic tensor---the
components that are zero for a monoclinic tensor. 
Then, the second counterclockwise branch of Diagram~(\ref{eq:CD2}) is performed as follows.
Applying the Backus~\cite{backus} average, we obtain (Bos et al.~\cite{bos})
\begin{equation*}
\langle c_{3333}\rangle=\overline{\left(\dfrac{1}{c_{3333}}\right)}^{\,\,-1}\,,
\qquad
\langle c_{2323}\rangle=\dfrac{\overline{\left(\dfrac{c_{2323}}{D}\right)}}{2D_2}\,,
\end{equation*} 
\begin{equation*}
\langle c_{1313}\rangle=\dfrac{\overline{\left(\dfrac{c_{1313}}{D}\right)}}{2D_2}\,,
\qquad
\langle c_{2313}\rangle=\dfrac{\overline{\left(\dfrac{c_{2313}}{D}\right)}}{2D_2}\,,
\end{equation*}
where $D\equiv 2(c_{2323}c_{1313}-c_{2313}^2)$ and
$D_2\equiv \overline{(c_{1313}/D)}\,\,\overline{(c_{2323}/D)}-\overline{(c_{2313}/D)}^{\,2}$\,.
We also obtain
\begin{equation*}
\langle c_{1133}\rangle=
\overline{\left(\dfrac{1}{c_{3333}}\right)}^{\,\,-1}
\overline{\left(\dfrac{c_{1133}}{c_{3333}}\right)}\,,
\quad
\langle c_{2233}\rangle=
\overline{\left(\dfrac{1}{c_{3333}}\right)}^{\,\,-1}
\overline{\left(\dfrac{c_{2233}}{c_{3333}}\right)}\,,
\end{equation*}
\begin{equation*}
\langle c_{3312}\rangle=
\overline{\left(\dfrac{1}{c_{3333}}\right)}^{\,\,-1}
\overline{\left(\dfrac{c_{3312}}{c_{3333}}\right)}\,,
\quad
\langle c_{1111}\rangle=
\overline{c_{1111}}-\overline{\left(\dfrac{c_{1133}^2}{c_{3333}}\right)}+
\overline{\left(\dfrac{1}{c_{3333}}\right)}^{\,\,-1}
\overline{\left(\dfrac{c_{1133}}{c_{3333}}\right)}^{\,2}\,,
\end{equation*}
\begin{equation*}
\langle c_{1122}\rangle=
\overline{c_{1122}}-\overline{\left(\dfrac{c_{1133}\,c_{2233}}{c_{3333}}\right)}+
\overline{\left(\dfrac{1}{c_{3333}}\right)}^{\,\,-1}
\overline{\left(\dfrac{c_{1133}}{c_{3333}}\right)}\,\,
\overline{\left(\dfrac{c_{2233}}{c_{3333}}\right)}\,,
\end{equation*}
\begin{equation*}
\langle c_{2222}\rangle=
\overline{c_{2222}}-\overline{\left(\dfrac{c_{2233}^2}{c_{3333}}\right)}+
\overline{\left(\dfrac{1}{c_{3333}}\right)}^{\,\,-1}
\overline{\left(\dfrac{c_{2233}}{c_{3333}}\right)}^{\,2}\,,
\end{equation*}
\begin{equation*}
\langle c_{1212}\rangle=
\overline{c_{1212}}-\overline{\left(\dfrac{c_{3312}^2}{c_{3333}}\right)}+
\overline{\left(\dfrac{1}{c_{3333}}\right)}^{\,\,-1}
\overline{\left(\dfrac{c_{3312}}{c_{3333}}\right)}^{\,2}\,,
\end{equation*}
\begin{equation*}
\langle c_{1112}\rangle=
\overline{c_{1112}}-\overline{\left(\dfrac{c_{3312}\,c_{1133}}{c_{3333}}\right)}+
\overline{\left(\dfrac{1}{c_{3333}}\right)}^{\,\,-1}
\overline{\left(\dfrac{c_{1133}}{c_{3333}}\right)}\,\,
\overline{\left(\dfrac{c_{3312}}{c_{3333}}\right)}
\end{equation*}
and
\begin{equation*}
\langle c_{2212}\rangle=
\overline{c_{2212}}-\overline{\left(\dfrac{c_{3312}\,c_{2233}}{c_{3333}}\right)}+
\overline{\left(\dfrac{1}{c_{3333}}\right)}^{\,\,-1}
\overline{\left(\dfrac{c_{2233}}{c_{3333}}\right)}\,\,
\overline{\left(\dfrac{c_{3312}}{c_{3333}}\right)}\,,
\end{equation*}
where angle brackets denote the equivalent-medium elasticity parameters.
The other equivalent-medium elasticity parameters are zero.

Following the clockwise path of Diagram~(\ref{eq:CD2}),
the upper branch is derived in matrix form in Bos et al.~\cite{bos}.
Then, in accordance with Bos et al.~\cite{bos}, the result of the right-hand branch is
derived by setting entries in the generally anisotropic tensor 
that are zero for a monoclinic tensor to zero.
The nonzero entries, which are too complicated to display explicitly, are---in general---not the same as the result of the counterclockwise path.
Hence, for generally anisotropic and monoclinic symmetries, the Backus~\cite{backus} and Gazis et al.~\cite{gazis} averages do not commute.
\end{proof}
\subsection{Monoclinic layers and orthotropic medium}
\label{sec:mono}
Theorem~\ref{thm:One} remains valid for layers exhibiting higher material symmetries.
For such symmetries, simpler expressions of the corresponding elasticity tensors allow us to examine special cases that result in commutativity.
Let us consider the following instance of Theorem~\ref{thm:One}.
\begin{proposition}
\label{thm:Two}
For the monoclinic and orthotropic symmetries, the Backus~\cite{backus} and Gazis et al.~\cite{gazis} averages do not commute.
\end{proposition}

To study this case, let us consider the following diagram,
\begin{equation}
\label{eq:CD}
\begin{CD}
\rm{mono}@>\rm{B}>>\rm{mono}\\
@V\mathrm{G}VV                         @VV\rm{G}V\\
\rm{ortho}@>>\rm{B}>\rm{ortho}
\end{CD}
\end{equation}
and the following lemma, whose proof is in Appendix~\ref{AppOne2}.
\begin{lemma}
\label{lem:Ortho}
For the effective orthotropic symmetry, the result of the Gazis et al.~\cite{gazis} average is tantamount to replacing each $c_{ijk\ell}$\,, in a generally an\-i\-so\-tro\-pic{---}or monoclinic---tensor, by its corresponding $c_{ijk\ell}$ of an orthotropic tensor, expressed in the natural coordinate system, including the replacements by the corresponding zeros. 
\end{lemma}
\noindent Lemma~\ref{lem:Ortho} entails a corollary.
\begin{corollary}
\label{col:Ortho}
For the effective orthotropic symmetry, given a generally anisotropic \hspace*{-0.05in}---or monoclinic---tensor,~$C$\,, 
\begin{equation}
\label{eq:GazisOrtho}
\widetilde{C}^{\,\rm ortho}=C^{\,\rm ortho}
\,.
\end{equation}
where $\widetilde{C}^{\,\rm ortho}$ is the Gazis et al.~\cite{gazis} average of~$C$\,, and $C^{\,\rm ortho}$ is an orthotropic tensor whose nonzero entries are the same as for~$C$\,.
\end{corollary}
\begin {proof} (of Proposition~\ref{thm:Two}) Let us consider a monoclinic tensor and proceed counterclockwise along the first branch of Diagram~(\ref{eq:CD}).
Using the fact that the monoclinic symmetry is a special case of general anisotropy, we invoke Corollary~\ref{col:Ortho} to conclude that $\widetilde{C}^{\,\rm ortho}=C^{\,\rm ortho}$\,,
which is equivalent to setting $c_{1112}$\,, $c_{2212}$\,, $c_{3312}$ and $c_{2313}$ to zero
in the monoclinic tensor.
We perform the upper branch of Diagram~(\ref{eq:CD}), which is the averaging of a stack of monoclinic layers to get a monoclinic equivalent medium, as in the case of the lower branch of Diagram~(\ref{eq:CD2}).
Thus, following the clockwise path, we obtain
\begin{equation}
c_{1212}^\circlearrowright=
\overline{c_{1212}}-\overline{\left(\dfrac{c_{3312}^2}{c_{3333}}\right)}+
\overline{\left(\dfrac{1}{c_{3333}}\right)}^{\,\,-1}
\overline{\left(\dfrac{c_{3312}}{c_{3333}}\right)}^{\,2}\,,
\label{eq:cl1}
\end{equation}
\begin{equation}
c_{1313}^\circlearrowright=\dfrac{\overline{\left(\dfrac{c_{1313}}{D}\right)}}{2D_2}\,,\qquad
c_{2323}^\circlearrowright=\dfrac{\overline{\left(\dfrac{c_{2323}}{D}\right)}}{2D_2}\,.
\label{eq:cl2}
\end{equation}
Following the counterclockwise path, we obtain
\begin{equation}
c_{1212}^\circlearrowleft=\overline{c_{1212}}\,,\quad
c_{1313}^\circlearrowleft=\overline{\left(\dfrac{1}{c_{1313}}\right)}^{\,\,-1}\,,\quad
c_{2323}^\circlearrowleft=\overline{\left(\dfrac{1}{c_{2323}}\right)}^{\,\,-1}\,.
\label{eq:ccl}
\end{equation}
The other entries are the same for both paths.

In conclusion, the results of the clockwise and  counterclockwise paths are the same if $c_{2313}=c_{3312}=0$\,, which is a special case of monoclinic symmetry.
Thus, the Backus~\cite{backus} average and Gazis et al.~\cite{gazis} average commute for that case, even though they do not in general. 
\end{proof}

Now, let us consider the case of weak anisotropy, in which $c_{2313}$ and $c_{3312}$\,,
which are zero for isotropy,  are small.
To study the commutativity of the two averages, consider the commutator, $\mathscr{C}=[B,G]=BG-GB$\,, where $BG$ is the clockwise path and $GB$ is the counterclockwise path.
Since---if $c_{2313}=c_{3312}=0$---the commutator is zero, it is to be expected that in a neighbourhood of this case we have near commutativity.
Specifically, if both  $c_{2313}$ and $c_{3312}$ are of order $\epsilon$\,, then $\mathscr{C}$ should also be of order $\epsilon$\,, which means that there is near commutativity up to this order.
However, remarkably, a much stronger statement is true.
It turns out that for  $c_{2313}$ and $c_{3312}$ of order $\epsilon$\,, $\mathscr{C}$ is of order $\epsilon^2$\,, thus indicating a much stronger near commutativity that could expected {\it a priori}.
This follows from the following Jacobian calculation.

In this case, $\mathscr{C}=[\mathscr{C}_1,\mathscr{C}_2,\mathscr{C}_3]$\,, where
\begin{equation*}
\mathscr{C}_1=c_{2323}^\circlearrowright-c_{2323}^\circlearrowleft
=\dfrac{\overline{\left(\dfrac{c_{2323}}{D}\right)}}{2D_2}
-
\overline{\left(\dfrac{1}{c_{2323}}\right)}^{\,\,-1}\,,
\end{equation*}
\begin{equation*}
\mathscr{C}_2=c_{1313}^\circlearrowright-c_{1313}^\circlearrowleft
=\dfrac{\overline{\left(\dfrac{c_{1313}}{D}\right)}}{2D_2}
-
\overline{\left(\dfrac{1}{c_{1313}}\right)}^{\,\,-1}
\end{equation*}
and
\begin{equation*}
\mathscr{C}_3=
c_{1212}^\circlearrowright-c_{1212}^\circlearrowleft
=
\overline{\left(\dfrac{1}{c_{3333}}\right)}^{\,\,-1}
\overline{\left(\dfrac{c_{3312}}{c_{3333}}\right)}^{\,2}
-\overline{\left(\dfrac{c_{3312}^2}{c_{3333}}\right)}\,.
\end{equation*}

The starting parameters are
\begin{equation*}
x=c_{3333}^i\,,c_{2323}^i\,,c_{1313}^i\,,c_{2313}^i\,,c_{3312}^i\,,\quad i=1,\ldots,n\,,
\end{equation*}
and we have commutativity if
\begin{equation*}
c_{2313}^i=c_{3312}^i=0\,,\quad  i=1,\ldots,n\,,
\end{equation*}
which we denote by $x=a$\,, such that $\mathscr{C}(a)=[0]$.

Let the average be the arithmetic average and assume that all layers have the same thickness, so that
\begin{equation*}
\overline{F}=\dfrac{1}{n}\sum_{i=1}^n F^i\,.
\end{equation*}
Also, we let the $3\times 5n$ Jacobian matrix be
\begin{equation*}
\mathscr{C}'(x)=\left[\dfrac{\partial\mathscr{C}}{\partial x}\right]\,.
\end{equation*}
In Appendix~\ref{app:jacobian} we evaluate this Jacobian and find that $\mathscr{C}'(a)=[0]$\,.

If we expand $\mathscr{C}(x)$ in a Taylor series,
\begin{equation}
\label{eq:TaylorCommutator}
\mathscr{C}(x)=\mathscr{C}(a)+\mathscr{C}'(a)(x-a) +\cdots=\mathscr{C}'(a)(x-a) +\cdots\,,
\end{equation}
then we see that, near $x=a$\,, $||\mathscr{C}(x)||=O\left(||x-a||^2\right)$\,, so that  there is very near  commutativity in a neighbourhood of $x=a$\,. 
In Section~\ref{sec:num} we illustrate numerically this strong near commutativity.
\subsection{Orthotropic layers and tetragonal medium}
\label{sec:ortho}
In a manner analogous to Diagram~(\ref{eq:CD}), but proceeding from the the upper-left-hand corner orthotropic tensor to lower-right-hand corner tetragonal tensor by the counterclockwise path,
\begin{equation}
\label{eq:CD3}
\begin{CD}
\rm{ortho}@>\rm{B}>>\rm{ortho}\\
@V\mathrm{G}VV                         @VV\rm{G}V\\
\rm{tetra}@>>\rm{B}>\rm{tetra}
\end{CD}
\end{equation}
we obtain
\begin{equation*}
c_{1111}^\circlearrowleft=\overline{\dfrac{c_{1111}+c_{2222}}{2}-
\dfrac{\left(\dfrac{c_{1111}+c_{2222}}{2}\right)^2}{c_{3333}}}+
\overline{\left(\dfrac{c_{1111}+c_{2222}}{2c_{3333}}\right)}^{\,2}
\,\overline{\left(\dfrac{1}{c_{3333}}\right)}^{\,\,-1}
\,.
\end{equation*}
Following the clockwise path, we obtain
\begin{equation*}
c_{1111}^\circlearrowright=\overline{\dfrac{c_{1111}+c_{2222}}{2}-
\dfrac{c_{1133}^2+c_{2233}^2}{2c_{3333}}}+
\dfrac{1}{2}\left[\overline{\left(\dfrac{c_{1133}}{c_{3333}}\right)}^{\,2}+
\overline{\left(\dfrac{c_{2233}}{c_{3333}}\right)}^{\,2}\right]
\overline{\left(\dfrac{1}{c_{3333}}\right)}^{\,\,-1}\,.
\end{equation*}
These results are not equal to one another, unless $c_{1133}=c_{2233}$\,, which is a special case of orthotropic symmetry.  The same is true for $c_{1122}^\circlearrowleft$ and $c_{1122}^\circlearrowright$.   
Also, $c_{2323}$ must equal $c_{1313}$ for $c_{2323}^\circlearrowright=c_{2323}^\circlearrowleft$.
The other entries are the same for both paths.
Thus, the Backus~\cite{backus} average and Gazis et al.~\cite{gazis} average do commute for $c_{1133}=c_{2233}$ and $c_{2323}=c_{1313}$\,, which is a special case of orthotropic symmetry, but they do not commute in general.

Similarly to our discussion in Section~\ref{sec:mono} and Appendix~\ref{app:jacobian}, we
examine the commutator.
Herein, the commutator is $\mathscr{C}=[\mathscr{C}_1,\mathscr{C}_2,\mathscr{C}_3]$\,, where
\begin{equation*}
\mathscr{C}_1=c_{1111}^\circlearrowright-c_{1111}^\circlearrowleft\,,
\quad
\mathscr{C}_2=c_{1122}^\circlearrowright-c_{1122}^\circlearrowleft\,,
\quad
\mathscr{C}_3=c_{2323}^\circlearrowright-c_{2323}^\circlearrowleft\,.
\end{equation*}
The starting parameters that show up in the commutator are
\begin{equation*}
x=c_{1133}^i\,,c_{2233}^i\,,c_{3333}^i\,,c_{2323}^i\,,c_{1313}^i\,,\quad i=1,\ldots,n\,,
\end{equation*}
and we have commutativity if
\begin{equation*}
c_{1133}^i=c_{2233}^i\quad{\rm and}\quad c_{2323}^i=c_{1313}^i\,,\quad  i=1,\ldots,n\,,
\end{equation*}
which we denote by $x=a$\,, such that $\mathscr{C}(a)=[0]$.

Again, as in Section~\ref{sec:mono}, we let the $3\times 5n$ Jacobian matrix be
\begin{equation*}
\mathscr{C}'(x)=\left[\dfrac{\partial\mathscr{C}}{\partial x}\right]\,.
\end{equation*}
In a series of calculations similar to those
in Appendix~\ref{app:jacobian} we evaluate this Jacobian and again find that $\mathscr{C}'(a)=[0]$\,.

Let us also examine the process of combining the Gazis et al.~\cite{gazis} averages, which is tantamount to combining Diagrams~(\ref{eq:CD}) and~(\ref{eq:CD3}),
\begin{equation}
\begin{CD}
\label{eq:CD4}
\rm{mono}@>\rm{B}>>\rm{mono}\\
@V\mathrm{G}VV                         @VV\rm{G}V\\
\rm{ortho}@>>\rm{B}>\rm{ortho}\\
@V\mathrm{G}VV                         @VV\rm{G}V\\
\rm{tetra}@>>\rm{B}>\rm{tetra}
\end{CD}
\end{equation}
In accordance with Theorem~\ref{thm:One}, in general, there is no commutativity.
However, the outcomes are the same as for the corresponding steps in Sections~\ref{sec:mono} and \ref{sec:ortho}.
In general, for the Gazis et al.~\cite{gazis} average, proceeding directly, $\rm{aniso}\xrightarrow{\rm{G}}\rm{iso}$\,, is tantamount to proceeding along arrows in Figure~\ref{fig:orderrelation}, $\rm{aniso}\xrightarrow{\rm{G}}\cdots\xrightarrow{\rm{G}}\rm{iso}$\,.
No such combining of the Backus~\cite{backus} averages is possible, since, for each step, layers become a homogeneous medium.
\subsection{Transversely isotropic layers}
Lack of commutativity between the two averages can be also exemplified by the case of transversely isotropic layers.
Following the clockwise path of Diagram~(\ref{eq:CD}), the Backus~\cite{backus} average results in a transversely isotropic medium, whose Gazis et al.~\cite{gazis} average---in accordance with Figure~\ref{fig:orderrelation}---is isotropic.
Following the counterclockwise path, Gazis et al.~\cite{gazis} average results in an isotropic medium, whose Backus~\cite{backus} average, however, is transverse isotropy.
Thus, not only the elasticity parameters, but even the resulting material-symmetry classes differ.

Also, we could---in a manner analogous to the one illustrated in Diagram~(\ref{eq:CD4})\,---begin with generally anisotropic layers and obtain isotropy by the clockwise path and transverse isotropy by the counterclockwise path, which again illustrates noncommutativity.
\section{Numerical examination}
\label{sec:num}
\subsection{Introduction}
In this section, we study numerically the extent of the lack of commutativity between the Backus~\cite{backus} and Gazis et al.~\cite{gazis} averages.
Also, we examine the effect of the strength of the anisotropy on noncommutativity.

We are once again dealing with Diagram~(\ref{eq:CD}).   Herein, $\rm B$ and $\rm G$ stand for the Backus~\cite{backus} average and the Gazis~et~al.~\cite{gazis} average, respectively.
The upper left-hand corner of Diagram~(\ref{eq:CD}) is a series of parallel monoclinic layers.
The lower right-hand corner is a single orthotropic medium.
The intermediate clockwise result is a single monoclinic tensor: an equivalent medium; the intermediate counterclockwise result is a series of parallel orthotropic layers: effective media.

As discussed in Section~\ref{sec:AnalForm}, even though, in general, the Backus~\cite{backus} average and the Gazis~et~al.~\cite{gazis} average do not commute, except in particular cases, it is important to consider the extent of their noncommutativity.
In other words, we enquire to what extent---in the context of a continuum-mechanics model and unavoidable measurement errors---the averages could be considered as approximately commutative.

To do so, we numerically examine two cases.
In one case, we begin---in the upper left-hand corner of Diagram~(\ref{eq:CD})---with ten strongly anisotropic layers.
In the other case, we begin with ten weakly anisotropic layers.
\subsection{Monoclinic layers and orthotropic medium}
\label{sub:MonoOrtho}
The elasticity parameters of the strongly anisotropic layers are derived by random variation of a feldspar given by Waeselmann et~al.~\cite{waeselmann}.
For consistency, we express these parameters in the natural coordinate system whose $x_3$-axis is perpendicular to the symmetry plane, as opposed to the~$x_2$-axis used by Waeselmann et~al.~\cite{waeselmann}.
These parameters are given in Table~\ref{tab:strong}.   

\begin{table}[H]
\caption{Ten strongly anisotropic monoclinic tensors.  The elasticity parameters are density-scaled; their units are $10^6~{\rm m}^2/{\rm s}^2$\,.}
\label{tab:strong}
{\normalsize
\setlength{\tabcolsep}{4pt}
\begin{tabular}{|c|c|c|c|c|c|c|c|c|c|c|c|c|c|}
\hline
layer&$c_{1111}$&$c_{1122}$&$c_{1133}$&$c_{1112}$&$c_{2222}$&$c_{2233}$&
$c_{2212}$&$c_{3333}$&$c_{3312}$&$c_{2323}$&$c_{2313}$&$c_{1313}$&$c_{1212}$\\
\hline
1&23.9&11.6&12.2&1.53&71.4&6.64&2.94&52.0&-2.89&8.00&-6.79&8.21&4.54\\
2&33.5&8.24&12.2&-0.98&66.9&5.65&2.02&82.3&-1.12&6.35&-5.16&17.4&7.36\\
3&33.2&9.79&16.9&0.57&62.1&6.19&3.81&83.4&-7.34&10.2&-2.33&16.6&4.72\\
4&38.1&8.33&12.2&1.51&55.0&4.87&3.11&56.8&-1.43&4.10&-0.20&8.25&11.2\\
5&37.4&11.5&14.4&-0.79&72.6&3.93&3.00&76.5&-6.07&9.58&-4.38&14.8&8.70\\
6&38.4&10.7&17.1&1.55&63.8&7.11&1.99&55.2&-0.98&9.66&-6.85&11.1&11.4\\
7&29.2&11.4&11.7&0.59&59.5&5.23&3.74&82.7&-3.81&10.1&-5.09&9.78&6.89\\
8&31.9&9.03&19.1&-0.07&71.6&4.18&1.98&70.4&-0.25&4.84&-0.33&8.21&10.9\\
9&37.5&10.5&19.4&0.37&76.7&5.02&3.57&76.7&-0.16&7.84&-1.62&13.8&10.7\\
10&36.0&9.65&18.9&-0.43&73.1&3.94&2.53&60.4&-7.20&5.44&-2.20&9.25&5.20\\
\hline
\end{tabular}}
\end{table}

The elasticity parameters of the weakly anisotropic layers are derived from the strongly anisotropic ones by keeping $c_{1111}$ and $c_{2323}$\,, which are the two distinct elasticity parameters of isotropy, approximately the same as for the corresponding strongly anisotropic layers, and by varying slightly other parameters away from isotropy.
These parameters are given in Table~\ref{tab:weak}.

\begin{table}[H]
\caption{Ten weakly anisotropic monoclinic tensors. The elasticity parameters are density-scaled; their units are $10^6~{\rm m}^2/{\rm s}^2$\,.}
\label{tab:weak}
{\normalsize
\setlength{\tabcolsep}{4pt}
\begin{tabular}{|c|c|c|c|c|c|c|c|c|c|c|c|c|c|}
\hline
layer&$c_{1111}$&$c_{1122}$&$c_{1133}$&$c_{1112}$&$c_{2222}$&$c_{2233}$&
$c_{2212}$&$c_{3333}$&$c_{3312}$&$c_{2323}$&$c_{2313}$&$c_{1313}$&$c_{1212}$\\
\hline
1&24&9&9&0.2&29&7&0.3&27&-0.3&8&-1&8.2&7\\
2&34&15&18&-0.1&38&14&0.2&39&-0.1&6&-1&7.5&6.5\\
3&33&12&14&0.06&37&10&0.4&38&-0.7&10&-0.5&12&8.5\\
4&38&20&22&0.15&40&15&0.3&41&-0.1&4&-0.2&5&6\\
5&37&14&16&-0.08&42&10&0.3&41&-0.6&10&-0.8&11&9\\
6&38&15&18&0.16&41&14&0.2&40&-0.1&10&-1&10.5&11\\
7&29&9.5&9.5&0.06&32&8&0.4&34&-0.4&10&-0.8&10&9\\
8&32&15&19.5&-0.01&36&13&0.2&36&-0.03&5&-0.3&6&6\\
9&38&16&20&0.04&43&14&0.4&42&-0.02&8&-0.4&9&9\\
10&36&18&23&-0.04&40&15&0.3&39&-0.7&5&-0.5&6&5\\
\hline
\end{tabular}}
\end{table}

Assuming that all layers have the same thickness, we use an arithmetic average for the Backus~\cite{backus} averaging; for instance,
\begin{equation*}
\overline{c_{1212}}=\dfrac{1}{10}\sum_{i=1}^{10} c_{1212}^i\,.
\end{equation*}

The results of the clockwise and counterclockwise paths 
for the three elasticity parameters that differ from each other are 
calculated from Equations~(\ref{eq:cl1}), (\ref{eq:cl2}) and~(\ref{eq:ccl}), and given in Table~\ref{tab:num}.
It appears that the averages nearly commute for the case of weak anisotropy.
Hence, we confirm, as discussed in Section~\ref{sec:mono}, that the extent of noncommutativity is a function of the strength of anisotropy.

\begin{table}[H]
\caption{Comparison of numerical results.}
\label{tab:num}
\centerline{
\begin{tabular}{||c||c|c||c|c||c|c||}
\hline
&&&&&&\\[-10pt]
anisotropy&$c_{1212}^\circlearrowright$&$c_{1212}^\circlearrowleft$&$c_{1313}^\circlearrowright$&$c_{1313}^\circlearrowleft$&$c_{2323}^\circlearrowright$&$c_{2323}^\circlearrowleft$\\
\hline
strong&8.06&8.16&9.13&10.84&6.36&6.90\\
weak&7.70&7.70&7.88&7.87&6.82&6.81\\
\hline
\end{tabular}
}
\end{table}

To ensure that our calculation of the Jacobian being zero is correct, as obtained in Section~\ref{sec:mono} and Appendix~\ref{app:jacobian}, we perform another test. 
We multiply the weakly anisotropic values of $c_{2313}^i$ and  $c_{3312}^i$\,, where $i=1,\ldots,n$\,, by~$\tfrac{1}{2}$ to find that, as expected, the commutator is multiplied by~$\tfrac{1}{4}$.

To quantify the strength of anisotropy, we invoke the concept of distance in the space of elasticity tensors (Danek~et ~al.~\cite{dks1,dks2}, Kochetov and Slawinski~\cite{ks1,ks2}).
In particular, we consider the closest isotropic tensor---according to the Frobenius norm---as formulated by Voigt~\cite{voigt}.
Examining one layer from the upper left-hand corner of Diagram~(\ref{eq:CD}), we denote its weakly anisotropic tensor as $c^{\rm w}$ and its strongly anisotropic tensor as $c^{\,\rm s}$\,.

Using explicit expressions of Slawinski~\cite{slawinski3}, we find that the elasticity parameters of the closest isotropic tensor,~$c^{\rm iso_w}$\,, to $c^{\rm w}$ is $c^{\rm iso_w}_{1111}=25.52$ and $c^{\rm iso_w}_{2323}=8.307$\,.
The Frobenius distance from $c^{\rm w}$ to $c^{\rm iso_w}$ is $6.328$\,.
The closest isotropic tensor,~$c^{\rm iso_s}$, to $c^{\,\rm s}$ is $c^{\rm iso_s}_{1111}=39.08$ and $c^{\rm iso_s}_{2323}=11.94$\,.
The distance from $c^{\,\rm s}$ to $c^{\rm iso_s}$ is $49.16$\,.

Thus, as expected, $c^{\,\rm s}$\,, which represents strong anisotropy,  is much further from isotropy than $c^{\rm w}$\,, which represents weak anisotropy.
\subsection{Orthotropic layers and tetragonal medium}
To examine further the commutativity of averages, we generate ten weakly anisotropic orthotropic tensors from the ten weakly
anisotropic monoclinic tensors by setting appropriate entries to zero.
Similarly to the weakly anisotropic case discussed in Section~\ref{sub:MonoOrtho}, we find that the Backus~\cite{backus} and Gazis et al.~\cite{gazis} averages nearly commute.

As shown in Section~\ref{sec:ortho}---for orthotropic layers and a tetragonal medium---there is commutativity only if
$c_{1133}^i=c_{2233}^i$ and $c_{2323}^i=c_{1313}^i$\,, which,  in this case, corresponds to $x=a$ in expression~(\ref{eq:TaylorCommutator}).   

If we multiply the difference between the weakly anisotropic values of $c_{1133}^i$ and $c_{2233}^i$ as well as that between  $c_{2323}^i$ and $c_{1313}^i$ by a factor of $F$\,, we find that $\mathscr{C}$ is multiplied by approximately a factor of $F^2$.
The factors of $F$ used in these examination are $\tfrac{1}{2}$\,, $\tfrac{1}{3}$\,, $\tfrac{1}{4}$ and
$\tfrac{1}{10}$\,, with nearly exact values of $F^2$ for $\mathscr{C}_1$ and $\mathscr{C}_2$ and a close value for~$\mathscr{C}_3$. 
Thus, again, if the differences are of order $\epsilon$\,, the commutator is of order $\epsilon^2$.
\section{Discussion}
We conclude that---in general---the Backus~\cite{backus} average, which is a spatial average over an inhomogeneity, and the Gazis et~al.~\cite{gazis} average, which is an average over an anisotropic symmetry group at a point, do not commute. 
Mathematically, this noncommutativity is stated by Pro\-position~\ref{thm:One}.
Also, it is exemplified for several material symmetries.

There are, however, particular cases of given material symmetries for which the averaging processes commute, as discussed in Sections~\ref{sec:mono} and \ref{sec:ortho}.
Yet, we do not see a physical explanation for the commutativity in these special cases, which is consistent with the view that a mathematical realm---even though it allows us to formulate quantitative analogies for the physical world---has no causal connection with it.

Using the the case of monoclinic and orthotropic symmetries, we numerically show that noncommutativity is a function of the strength of anisotropy.
For weak anisotropy, which is a common case of seismological studies, the averages nearly commute.
Furthermore, and perhaps surprisingly, a perturbation of the elasticity parameters about a point of weak anisotropy results in the commutator of the two types of averaging being of the order of the square of this perturbation.

For theoretical seismology, which is our motivation, weak anisotropy is
adequate for most cases; hence, this near commutativity is welcome. In other words,
the fact that the order of a sequence of these two averages is nearly
indistinguishable is important information.

In this study---for convenience and without appreciable loss of generality---we assume that all tensors are expressed in the same orientation of their coordinate systems.
Otherwise, the process of averaging become more complicated, as discussed---for the Gazis et al.~\cite{gazis} average---by Kochetov and Slawinski~\cite{ks1,ks2} and as mentioned---for the Backus~\cite{backus} average---by Bos et al.~\cite{bos}.
\section*{Acknowledgements}
We wish to acknowledge discussions with Theodore Stanoev.
The numerical examination was motivated by a discussion with Robert Sarracino.
This research was performed in the context of The Geomechanics Project supported by Husky Energy.
Also, this research was partially supported by the Natural Sciences and Engineering Research Council of Canada, grant 238416-2013.
\bibliographystyle{unsrt}
\bibliography{commutative}

\begin{thebibliography}{10}

\bibitem{backus}
G.~E. Backus.
\newblock Long-wave elastic anisotropy produced by horizontal layering.
\newblock {\em J. Geophys. Res.}, 67(11):4427--4440, 1962.

\bibitem{gazis}
D.~C. Gazis, I.~Tadjbakhsh, and R.~A. Toupin.
\newblock The elastic tensor of given symmetry nearest to an anisotropic
  elastic tensor.
\newblock {\em Acta Crystallographica}, 16(9):917--922, 1963.

\bibitem{bos}
L.~Bos, D.~R. Dalton, M.~A. Slawinski, and T.~Stanoev.
\newblock On {B}ackus average for generally anisotropic layers.
\newblock {\em Journal of Elasticity}, 127(2):179--196, 2017.

\bibitem{BosX}
L.~Bos, T.~Danek, M.~A. Slawinski, and T.~Stanoev.
\newblock Statistical and numerical considerations of {B}ackus-average product
  approximation.
\newblock {\em Journal of Elasticity}, DOI 10.1007/s10659-017-9659-9:1--16,
  2017.

\bibitem{dks1}
T.~Danek, M.~Kochetov, and M.~A. Slawinski.
\newblock Uncertainty analysis of effective elasticity tensors using
  quaternion-based global optimization and monte-carlo method.
\newblock {\em The Quarterly Journal of Mechanics and Applied Mathematics},
  66(2):253--272, 2013.

\bibitem{dks2}
T.~Danek, M.~Kochetov, and M.~A. Slawinski.
\newblock Effective elasticity tensors in the context of random errors.
\newblock {\em Journal of Elasticity}, 121(1):55--67, 2015.

\bibitem{waeselmann}
N.~Waeselmann, J.~M. Brown, R.~J. Angel, N.~Ross, J.~Zhao, and W.~Kamensky.
\newblock The elastic tensor of monoclinic alkali feldspars.
\newblock {\em American Mineralogist}, 101:1228--1231, 2016.

\bibitem{ks1}
M.~Kochetov and M.~A. Slawinski.
\newblock On obtaining effective orthotropic elasticity tensors.
\newblock {\em The Quarterly Journal of Mechanics and Applied Mathematics},
  62(2):149--166, 2009a.

\bibitem{ks2}
M.~Kochetov and M.~A. Slawinski.
\newblock On obtaining effective transversely isotropic elasticity tensors.
\newblock {\em Journal of Elasticity}, 94:1--13, 2009b.

\bibitem{voigt}
W.~Voigt.
\newblock {\em {L}ehrbuch der {K}ristallphysik}.
\newblock Teubner, Leipzig, 1910.

\bibitem{slawinski3}
M.~A. Slawinski.
\newblock {\em Waves and rays in seismology: {A}nswers to unasked questions}.
\newblock World Scientific, 2016.

\bibitem{thomson}
W.~Thomson.
\newblock {\em Mathematical and physical papers: {E}lasticity, heat,
  electromagnetism}.
\newblock Cambridge University Press, 1890.

\bibitem{chapman}
C.~H. Chapman.
\newblock {\em Fundamentals of seismic wave propagation}.
\newblock Cambridge University Press, 2004.

\bibitem{slawinski1}
M.~A. Slawinski.
\newblock {\em Waves and rays in elastic continua}.
\newblock World Scientific, 3rd edition, 2015.

\bibitem{bona}
A.~B\'{o}na, I.~Bucataru, and M.~A. Slawinski.
\newblock Space of ${SO}(3)$-orbits of elasticity tensors.
\newblock {\em Archives of Mechanics}, 60(2):123--138, 2008.

\end{thebibliography}
\setcounter{section}{0}
\renewcommand{\thesection}{\Alph{section}}
\section{Proofs of Lemmas}
\subsection{Lemma~\ref{lem:Mono}}\label{AppOne1}
\begin{proof}
For discrete symmetries, we can write integral~(\ref{eq:proj}) as a sum,
\begin{equation}
\label{eq:AverageDisc}
\widetilde C^{\,\rm sym}=\dfrac{1}{n}\left(\tilde{A}_1^{\rm sym}\,C\,\tilde{A}_1^{\rm sym}\,{}^{^T}+\ldots+\tilde{A}_n^{\rm sym}\,C\,\tilde{A}_n^{\rm sym}\,{}^{^T}\right)
\,,
\end{equation}
where $\widetilde C^{\rm sym}$ is expressed in Kelvin's notation, in view
of  Thomson~\cite[p.~110]{thomson}, as discussed in Chapman~\cite[Section~4.4.2]{chapman}.

To write the elements of the monoclinic symmetry group as $6\times 6$ matrices, we must consider orthogonal transformations in $\mathbb{R}^3$\,.
Transformation $A\in SO(3)$ of $c_{ijk\ell}$ corresponds to transformation of $C$ given by
\begin{align}
\tilde{A}&=\left[\begin{array}{cccc}
A_{11}^{2} & A_{12}^{2} & A_{13}^{2} & \sqrt{2}A_{12}A_{13}\\
A_{21}^{2} & A_{22}^{2} & A_{23}^{2} & \sqrt{2}A_{22}A_{23}\\
A_{31}^{2} & A_{32}^{2} & A_{33}^{2} & \sqrt{2}A_{32}A_{33}\\
\sqrt{2}A_{21}A_{31} & \sqrt{2}A_{22}A_{32} & \sqrt{2}A_{23}A_{33} & A_{23}A_{32}+A_{22}A_{33}\\
\sqrt{2}A_{11}A_{31} & \sqrt{2}A_{12}A_{32} & \sqrt{2}A_{13}A_{33} & A_{13}A_{32}+A_{12}A_{33}\\
\sqrt{2}A_{11}A_{21} & \sqrt{2}A_{12}A_{22} & \sqrt{2}A_{13}A_{23} & A_{13}A_{22}+A_{12}A_{23}\end{array}\right.\label{eq:ATildeQ}
\\
&\hspace*{1.5in}
\left.
\begin{array}{cc}
\sqrt{2}A_{11}A_{13} & \sqrt{2}A_{11}A_{12}\\
\sqrt{2}A_{21}A_{23} & \sqrt{2}A_{21}A_{22}\\
\sqrt{2}A_{31}A_{33} & \sqrt{2}A_{31}A_{32}\\
A_{23}A_{31}+A_{21}A_{33} & A_{22}A_{31}+A_{21}A_{32}\\
A_{13}A_{31}+A_{11}A_{33} & A_{12}A_{31}+A_{11}A_{32}\\
A_{13}A_{21}+A_{11}A_{23} & A_{12}A_{21}+A_{11}A_{22}
\end{array}
\right]
\,,\nonumber
\end{align}
which is an orthogonal matrix, $\tilde{A}\in SO(6)$ (Slawinski~\cite[Section~5.2.5]{slawinski1}).\footnote{Readers interested in formulation of matrix~(\ref{eq:ATildeQ}) might also refer to  B\'ona et al.~\cite{bona}.}

The required symmetry-group elements are
\begin{equation*}
 A_1^{\rm mono}=
\left[
\begin{array}{ccc}
1 & 0 & 0\\
0 & 1 & 0\\
0 & 0 & 1\end{array}\right]
\mapsto
\left[\begin{array}{cccccc}
1 & 0 & 0 & 0 & 0 & 0\\
0 & 1 & 0 & 0 & 0 & 0\\
0 & 0 & 1 & 0 & 0 & 0\\
0 & 0 & 0 & 1 & 0 & 0\\
0 & 0 & 0 & 0 & 1 & 0\\
0 & 0 & 0 & 0 & 0 & 1\end{array}\right]
=\tilde{A}_1^{\rm mono}
\end{equation*}
and
\begin{equation*}
 A_2^{\rm mono}=
\left[
\begin{array}{ccc}
-1 & 0 & 0\\
0 & -1 & 0\\
0 & 0 & 1\end{array}\right]
\mapsto
\left[\begin{array}{cccccc}
1 & 0 & 0 & 0 & 0 & 0\\
0 & 1 & 0 & 0 & 0 & 0\\
0 & 0 & 1 & 0 & 0 & 0\\
0 & 0 & 0 & -1 & 0 & 0\\
0 & 0 & 0 & 0 & -1 & 0\\
0 & 0 & 0 & 0 & 0 & 1\end{array}\right]
=\tilde{A}_2^{\rm mono}
\,.
\end{equation*}
For the monoclinic case, expression~(\ref{eq:AverageDisc}) can be stated explicitly as
\begin{equation*}
\widetilde C^{\rm mono}=
\dfrac{\left(\tilde{A}_1^{\rm mono}\right)\,C\,\left(\tilde{A}_1^{\rm mono}\right)^T+\left(\tilde{A}_2^{\rm mono}\right)\,C\,\left(\tilde{A}_2^{\rm mono}\right)^T}{2}
\,.
\end{equation*}
Performing matrix operations, we obtain
\begin{equation}
\widetilde C^{\rm mono}
=\left[\begin{array}{cccccc}
c_{1111} & c_{1122} & c_{1133} & 0 & 0 & \sqrt{2}c_{1112}\\
c_{1122} & c_{2222} & c_{2233} & 0 & 0 & \sqrt{2}c_{2212}\\
c_{1133} & c_{2233} & c_{3333} & 0 & 0 & \sqrt{2}c_{3312}\\
0 & 0 & 0 & 2c_{2323} & 2c_{2313} & 0\\
0 & 0 & 0 & 2c_{2313} & 2c_{1313} & 0\\
\sqrt{2}c_{1112} & \sqrt{2}c_{2212} & \sqrt{2}c_{3312} & 0 & 0 & 2c_{1212}
\end{array}\right]
\,,
\label{eq:MonoExplicitRef}
\end{equation}
which exhibits the form of the monoclinic tensor in its natural coordinate system.
In other words, $\widetilde{C}^{\rm mono}=C^{\rm mono}$\,, in accordance with Corollary~\ref{col:Mono}.
\end{proof}
\subsection{Lemma~\ref{lem:Ortho}}\label{AppOne2}
\begin{proof}
For orthotropic symmetry, \par\noindent $\tilde{A}_1^{\rm ortho}=\tilde{A}_1^{\rm mono}$\,, ${\tilde{A}_2^{\rm ortho}=\tilde{A}_2^{\rm mono}}$\,,
\begin{equation*}
 A_3^{\rm ortho}=
\left[
\begin{array}{ccc}
-1 & 0 & 0\\
0 &  1 & 0\\
0 & 0 & -1\end{array}\right]
\mapsto
\left[\begin{array}{cccccc}
1 & 0 & 0 & 0 & 0 & 0\\
0 & 1 & 0 & 0 & 0 & 0\\
0 & 0 & 1 & 0 & 0 & 0\\
0 & 0 & 0 & -1 & 0 & 0\\
0 & 0 & 0 & 0 &  1 & 0\\
0 & 0 & 0 & 0 & 0 & -1\end{array}\right]
=\tilde{A}_3^{\rm ortho}
\,,
\end{equation*}
and
\begin{equation*}
 A_4^{\rm ortho}=
\left[
\begin{array}{ccc}
1 & 0 & 0\\
0 &  -1 & 0\\
0 & 0 & -1\end{array}\right]
\mapsto
\left[\begin{array}{cccccc}
1 & 0 & 0 & 0 & 0 & 0\\
0 & 1 & 0 & 0 & 0 & 0\\
0 & 0 & 1 & 0 & 0 & 0\\
0 & 0 & 0 & 1 & 0 & 0\\
0 & 0 & 0 & 0 &  -1 & 0\\
0 & 0 & 0 & 0 & 0 & -1\end{array}\right]
=\tilde{A}_4^{\rm ortho}
\,.
\end{equation*}
For the orthotropic case, expression~(\ref{eq:AverageDisc}) can be stated explicitly as
\begin{align*}
\widetilde C^{\rm ortho}=&
\left[\left(\tilde{A}_1^{\rm ortho}\right)\,C\,\left(\tilde{A}_1^{\rm ortho}\right)^T
+\left(\tilde{A}_2^{\rm ortho}\right)\,C\,\left(\tilde{A}_2^{\rm ortho}\right)^T\right.\\
&+\left.\left(\tilde{A}_3^{\rm ortho}\right)\,C\,\left(\tilde{A}_3^{\rm ortho}\right)^T+\left(\tilde{A}_4^{\rm ortho}\right)\,C\,\left(\tilde{A}_4^{\rm ortho}\right)^T\right]/4
\,.
\end{align*}
Performing matrix operations, we obtain
\begin{equation}
\widetilde C^{\rm ortho}
=\left[\begin{array}{cccccc}
c_{1111} & c_{1122} & c_{1133} & 0 & 0 & 0\\
c_{1122} & c_{2222} & c_{2233} & 0 & 0 & 0\\
c_{1133} & c_{2233} & c_{3333} & 0 & 0 &0\\
0 & 0 & 0 & 2c_{2323} & 0 & 0\\
0 & 0 & 0 & 0 & 2c_{1313} & 0\\
0& 0 & 0 & 0 & 0 & 2c_{1212}
\end{array}\right]
\,,
\label{eq:OrthoExplicitRef}
\end{equation}
which exhibits the form of the orthotropic tensor in its natural coordinate system.
In other words, $\widetilde{C}^{\rm ortho}=C^{\rm ortho}$\,, in accordance with Corollary~\ref{col:Ortho}.
\end{proof}
\section{Evaluation of Jacobian}
\label{app:jacobian}
\begin{equation*}
\mathscr{C}_3=\left[\dfrac{1}{n}\sum\limits_{i=1}^n\dfrac{1}{c_{3333}^i}\right]^{-1}
\left[\dfrac{1}{n}\sum\limits_{i=1}^n\dfrac{c_{3312}^i}{c_{3333}^i}\right]^{2}
-
\left[\dfrac{1}{n}\sum\limits_{i=1}^n\dfrac{(c_{3312}^i)^2}{c_{3333}^i}\right]\,.
\end{equation*}
\begin{equation*}
\dfrac{\partial \mathscr{C}_3}{\partial c_{2323}^{\,j}}=\dfrac{\partial \mathscr{C}_3}{\partial c_{1313}^{\,j}}
=\dfrac{\partial \mathscr{C}_3}{\partial c_{2313}^{\,j}}=0\,.
\end{equation*}
\begin{equation*}
\dfrac{\partial \mathscr{C}_3}{\partial c_{3312}^{\,j}}=2\left[\dfrac{1}{n}\sum\limits_{i=1}^n\dfrac{1}{c_{3333}^i}\right]^{-1}\left[\dfrac{1}{n}\sum\limits_{i=1}^n\dfrac{c_{3312}^i}{c_{3333}^i}\right]\left(\dfrac{1}{n}\right)\left(\dfrac{1}{c_{3333}^{\,j}}\right)-\dfrac{2}{n}\dfrac{c_{3312}^{\,j}}{c_{3333}^{\,j}}\,.
\end{equation*}
\begin{align*}
\dfrac{\partial \mathscr{C}_3}{\partial c_{3333}^{\,j}}&=-\left[\dfrac{1}{n}\sum\limits_{i=1}^n\dfrac{1}{c_{3333}^i}\right]^{-2}
\left[\dfrac{1}{n}\left(\dfrac{-1}{(c_{3333}^{\,j})^2}\right)\right]\left[\dfrac{1}{n}\sum\limits_{i=1}^n\dfrac{c_{3312}^i}{c_{3333}^i}\right]^{2}\\
&\quad+
\left[\dfrac{1}{n}\sum\limits_{i=1}^n\dfrac{1}{c_{3333}^i}\right]^{-1}\left[\dfrac{2}{n}\sum\limits_{i=1}^n\dfrac{c_{3312}^i}{c_{3333}^i}\right]\left[\dfrac{1}{n}\left(\dfrac{-c_{3312}^{\,j}}{(c_{3333}^{\,j})^2}\right)\right]
+
\dfrac{1}{n}\dfrac{\left(c_{3312}^{\,j}\right)^2}{\left(c_{3333}^{\,j}\right)^2}
\,.
\end{align*}
Examining the above two equations---where for $x=a$\,, $c_{2313}^{\,j}=c_{3312}^{\,j}=0$\,, with $j=1,\ldots,n$---we see that
\begin{equation*}
\left.\dfrac{\partial \mathscr{C}_3}{\partial c_{3312}^{\,j}}\right|_{x=a}
=\left.\dfrac{\partial \mathscr{C}_3}{\partial c_{3333}^{\,j}}\right|_{x=a}=0\,.
\end{equation*}
Next, let us examine $\mathscr{C}_1$ and $\mathscr{C}_2$\,.
First, note that
\begin{equation*}
\dfrac{\partial \mathscr{C}_1}{\partial c_{3333}^{\,j}}=\dfrac{\partial \mathscr{C}_2}{\partial c_{3333}^{\,j}}
=\dfrac{\partial \mathscr{C}_1}{\partial c_{3312}^{\,j}}=\dfrac{\partial \mathscr{C}_2}{\partial c_{3312}^{\,j}}=0\,.
\end{equation*}
We let
\begin{equation*}
f=\dfrac{1}{n}\sum\limits_{i=1}^n\dfrac{c_{2323}^i}{2\left(c_{2323}^i c_{1313}^i-\left[c_{2313}^i\right]^2\right)}\,,
\end{equation*}
\begin{equation*}
g=\dfrac{1}{n}\sum\limits_{i=1}^n\dfrac{c_{1313}^i}{2\left(c_{2323}^i c_{1313}^i-\left[c_{2313}^i\right]^2\right)}
\end{equation*}
and
\begin{equation*}
h=\dfrac{1}{n}\sum\limits_{i=1}^n\dfrac{c_{2313}^i}{2\left(c_{2323}^i c_{1313}^i-\left[c_{2313}^i\right]^2\right)}\,.
\end{equation*}
which leads to
\begin{equation*}
\mathscr{C}_1=\dfrac{f}{2\left[fg-h^2\right]}-\left(\dfrac{1}{n}\sum\limits_{i=1}^n\dfrac{1}{c_{2323}^i}\right)^{-1}\,,
\end{equation*}
Thus,
\begin{align*}
\dfrac{\partial \mathscr{C}_1}{\partial c_{2323}^{\,j}}=\dfrac{\partial f}{\partial c_{2323}^{\,j}}\dfrac{1}
{2\left[fg-h^2\right]}&-\dfrac{f}{2}\left[fg-h^2\right]^{-2}\left[g\dfrac{\partial f}{\partial c_{2323}^{\,j}}+
f\dfrac{\partial g}{\partial c_{2323}^{\,j}}-2h\dfrac{\partial h}{\partial c_{2323}^{\,j}}\right]\\
&+\left(\dfrac{1}{n}\sum\limits_{i=1}^n\dfrac{1}{c_{2323}^i}\right)^{-2}\dfrac{1}{n}\dfrac{-1}{\left[c_{2323}^{\,j}\right]^2}\,,
\end{align*}
\begin{equation*}
\dfrac{\partial f}{\partial c_{2323}^{\,j}}=
\dfrac{1}{n}\dfrac{1}{2\left(c_{2323}^{\,j} c_{1313}^{\,j} -\left [c_{2313}^{\,j}\right]^2\right)}
-\dfrac{c_{2323}^{\,j}}{2n}c_{1313}^{\,j}\left(c_{2323}^{\,j} c_{1313}^{\,j} - \left[c_{2313}^{\,j}\right]^2\right)^{-2}\,,
\end{equation*}
\begin{equation*}
\dfrac{\partial g}{\partial c_{2323}^{\,j}}=
\dfrac{-\left[c_{1313}^{\,j}\right]^2}{2n}\left(c_{2323}^{\,j} c_{1313}^{\,j} - \left[c_{2313}^{\,j}\right]^2\right)^{-2}\,,
\end{equation*}
\begin{equation*}
\dfrac{\partial h}{\partial c_{2323}^{\,j}}=
\dfrac{c_{2313}^{\,j}}{2n}\left(2c_{1313}^{\,j}\right)\left(c_{2323}^{\,j} c_{1313}^{\,j} - \left[c_{2313}^{\,j}\right]^2\right)^{-2}\,.
\end{equation*}
\begin{equation*}
\left.\dfrac{\partial f}{\partial c_{2323}^{\,j}}\right|_{x=a}=
\dfrac{1}{2n c_{2323}^{\,j} c_{1313}^{\,j}}
-\dfrac{c_{2323}^{\,j} c_{1313}^{\,j}}{2n\left(c_{2323}^{\,j} c_{1313}^{\,j}\right)^2}=0\,.
\end{equation*}
\begin{equation*}
\left.\dfrac{\partial g}{\partial c_{2323}^{\,j}}\right|_{x=a}=
\dfrac{-\left(c_{1313}^{\,j}\right)^2}{2n\left(c_{2323}^{\,j}\right)^2\left(c_{1313}^{\,j}\right)^2}=\dfrac{-1}{2n\left(c_{2323}^{\,j}\right)^2}\,.
\end{equation*}
\begin{equation*}
\left.\dfrac{\partial h}{\partial c_{2323}^{\,j}}\right|_{x=a}=0\,.
\end{equation*}
\begin{equation*}
\left.\dfrac{\partial \mathscr{C}_1}{\partial c_{2323}^{\,j}}\right|_{x=a}=
0-\dfrac{f}{2}\left[fg-h^2\right]^{-2}\left[0-\dfrac{f}{2n(c_{2323}^{\,j})^2}-0\right]
+\left(\dfrac{1}{n}\sum\limits_{i=1}^n \dfrac{1}{c_{2323}^i}\right)^{-2}\dfrac{1}{n}\dfrac{-1}{\left(c_{2323}^{\,j}\right)^2}\,.\end{equation*}
\begin{equation*}
\left.\dfrac{f^2}{4n\left[fg-h^2\right]^2}\right|_{x=a}=
\dfrac{\left[\dfrac{1}{n}\sum\limits_{i=1}^n \left(\dfrac{1}{2c_{1313}^i}\right)\right]^2}
{4n\left(\dfrac{1}{4n^2}\sum\limits_{i=1}^n \dfrac{1}{c_{1313}^i}\sum\limits_{i=1}^n \dfrac{1}{c_{2323}^i}\right)^2}
=\dfrac{n}{\left(\sum\limits_{i=1}^n\dfrac{1}{c_{2323}^i}\right)^2}\,.
\end{equation*}
So,
\begin{equation*}
\left.\dfrac{\partial \mathscr{C}_1}{\partial c_{2323}^{\,j}}\right|_{x=a}=
\dfrac{n}{\left(\sum\limits_{i=1}^n \dfrac{1}{c_{2323}^i}\right)^2 \left(c_{2323}^{\,j}\right)^2}
-\dfrac{n}{\left(\sum\limits_{i=1}^n \dfrac{1}{c_{2323}^i}\right)^2 \left(c_{2323}^{\,j}\right)^2}=0\,.
\end{equation*}
Similarly, by symmetry of the equations,
\begin{equation*}
\left.\dfrac{\partial \mathscr{C}_2}{\partial c_{1313}^{\,j}}\right|_{x=a}=0\,.
\end{equation*}
Next, we consider the derivative with respect to $c_{1313}^{\,j}$.
\begin{equation*}
\dfrac{\partial \mathscr{C}_1}{\partial c_{1313}^{\,j}}=\dfrac{\partial f}{\partial c_{1313}^{\,j}}\dfrac{1}
{2\left[fg-h^2\right]}-\dfrac{f}{2}\left[fg-h^2\right]^{-2}\left[g\dfrac{\partial f}{\partial c_{1313}^{\,j}}+
f\dfrac{\partial g}{\partial c_{1313}^{\,j}}-2h\dfrac{\partial h}{\partial c_{1313}^{\,j}}\right]\,.
\end{equation*}
\begin{equation*}
\dfrac{\partial f}{\partial c_{1313}^{\,j}}=
\dfrac{-\left[c_{2323}^{\,j}\right]^2}{2n}\left(c_{2323}^{\,j} c_{1313}^{\,j} - \left[c_{2313}^{\,j}\right]^2\right)^{-2}\,,
\end{equation*}
\begin{equation*}
\dfrac{\partial g}{\partial c_{1313}^{\,j}}=
\dfrac{1}{n}\dfrac{1}{2\left(c_{2323}^{\,j} c_{1313}^{\,j} - \left[c_{2313}^{\,j}\right]^2\right)}
-\dfrac{c_{1313}^{\,j}}{2n}c_{2323}^{\,j}\left(c_{2323}^{\,j} c_{1313}^{\,j} - \left[c_{2313}^{\,j}\right]^2\right)^{-2}\,,
\end{equation*}
\begin{equation*}
\dfrac{\partial h}{\partial c_{1313}^{\,j}}=
\dfrac{c_{2313}^{\,j}}{2n}\left(2c_{2323}^{\,j}\right)\left(c_{2323}^{\,j} c_{1313}^{\,j} - \left[c_{2313}^{\,j}\right]^2\right)^{-2}\,.
\end{equation*}
These lead to
\begin{equation*}
\left.\dfrac{\partial f}{\partial c_{1313}^{\,j}}\right|_{x=a}=\dfrac{-1}{2n\left(c_{1313}^{\,j}\right)^2}\,,
\end{equation*}
\begin{equation*}
\left.\dfrac{\partial g}{\partial c_{1313}^{\,j}}\right|_{x=a}=
\dfrac{1}{2n c_{2323}^{\,j} c_{1313}^{\,j}}
-\dfrac{c_{1313}^{\,j} c_{2323}^{\,j}}{2n\left(c_{2323}^{\,j} c_{1313}^{\,j}\right)^2}=0\,,
\end{equation*}
\begin{equation*}
\left.\dfrac{\partial h}{\partial c_{1313}^{\,j}}\right|_{x=a}=0\,.
\end{equation*}
So,
\begin{equation*}
\left.\dfrac{\partial \mathscr{C}_1}{\partial c_{1313}^{\,j}}\right|_{x=a}=\left.
\dfrac{\partial f}{\partial c_{1313}^{\,j}}\dfrac{1}{2\left[fg-h^2\right]}
\left[1-\dfrac{fg}{fg-h^2}\right]\right|_{x=a}=0
\end{equation*}
and, similarly,
\begin{equation*}
\left.\dfrac{\partial \mathscr{C}_2}{\partial c_{2323}^{\,j}}\right|_{x=a}=0\,.
\end{equation*}
Next, we consider the derivative with respect to $c_{2313}^{\,j}$.
\begin{equation*}
\dfrac{\partial \mathscr{C}_1}{\partial c_{2313}^{\,j}}=\dfrac{\partial f}{\partial c_{2313}^{\,j}}\dfrac{1}
{2\left[fg-h^2\right]}-\dfrac{f}{2}\left[fg-h^2\right]^{-2}\left[g\dfrac{\partial f}{\partial c_{2313}^{\,j}}+
f\dfrac{\partial g}{\partial c_{2313}^{\,j}}-2h\dfrac{\partial h}{\partial c_{2313}^{\,j}}\right]\,.
\end{equation*}
\begin{equation*}
\left.\dfrac{\partial  f}{\partial c_{2313}^{\,j}}\right|_{x=a}=
\left.\dfrac{-2 c_{2313}^{\,j} c_{2323}^{\,j}}{2n\left(c_{2323}^{\,j} c_{1313}^{\,j} -\left[c_{2313}^{\,j}\right]^2\right)^2}\right|_{x=a}=0\,,
\end{equation*}
\begin{equation*}
\left.\dfrac{\partial  g}{\partial c_{2313}^{\,j}}\right|_{x=a}=
\left.\dfrac{-2 c_{2313}^{\,j} c_{1313}^{\,j}}{2n\left(c_{2323}^{\,j} c_{1313}^{\,j} -\left[c_{2313}^{\,j}\right]^2\right)^2}\right|_{x=a}=0\,,
\end{equation*}
\begin{align*}
\left.\dfrac{\partial  h}{\partial c_{2313}^{\,j}}\right|_{x=a}&=
\left.\dfrac{1}{2n\left(c_{2323}^{\,j} c_{1313}^{\,j} -\left[c_{2313}^{\,j}\right]^2\right)}\right|_{x=a}+
\left.\dfrac{2 (c_{2313}^{\,j})^2 }{2n\left(c_{2323}^{\,j} c_{1313}^{\,j} -\left[c_{2313}^{\,j}\right]^2\right)^2}
\right|_{x=a}\\&=\dfrac{1}{2n c_{2323}^{\,j} c_{1313}^{\,j}}\,.
\end{align*}
Thus,
\begin{equation*}
\left.\dfrac{\partial \mathscr{C}_1}{\partial c_{2313}^{\,j}}\right|_{x=a}=
0-\dfrac{f}{2}\left[fg-h^2\right]^{-2}\left[0+0-0\right]=0\,,
\end{equation*}
and, similarly,
\begin{equation*}
\left.\dfrac{\partial\mathscr{C}_2}{\partial c_{2313}^{\,j}}\right|_{x=a}=0\,.
\end{equation*}
Hence, $\mathscr{C}'(a)=[0]$\,; the Jacobian matrix is zero.
\end{document}